\documentclass[10pt]{article}

\usepackage[english]{babel}
\usepackage{amsmath}
\usepackage{amsthm}
\usepackage[shortlabels]{enumitem}
\usepackage[title]{appendix}
\usepackage{amsfonts}
\usepackage{mathtools}

\usepackage[ruled,vlined]{algorithm2e}
\usepackage{mathrsfs}
\usepackage{caption}
\usepackage[capbesideposition=outside,capbesidesep=quad]{floatrow}
\usepackage{tabu}
\usepackage{booktabs}
\usepackage{lipsum}
\usepackage{url}
\usepackage{subcaption}
\DeclareSymbolFont{bbold}{U}{bbold}{m}{n}
\DeclareSymbolFontAlphabet{\mathbbold}{bbold}
\usepackage{graphicx,xcolor} 
\usepackage[framemethod=tikz]{mdframed}
\usepackage{csquotes}
   \usepackage{tabu}
   \allowdisplaybreaks

\newtheorem{lemma}{Lemma}[section]
\newtheorem{proposition}{Proposition}[section]

\newcommand{\appropto}{\mathrel{\vcenter{
  \offinterlineskip\halign{\hfil$##$\cr
    \propto\cr\noalign{\kern2pt}\sim\cr\noalign{\kern-2pt}}}}}

\usepackage{authblk}

\usepackage{hyperref}

\usepackage{blindtext}

\newcommand\T{\rule{0pt}{3ex}}       

\title{The Nonuniversality of Wealth Distribution Tails\\ Near Wealth Condensation Criticality}

\author{Sam L. Polk\thanks{Corresponding Author. Present Address: 503 Boston Avenue, Medford, MA, USA. \newline \indent \; \textit{Email Addresses}: \url{Samuel.Polk@Tufts.edu} (Sam L. Polk), \url{Bruce.Boghosian@tufts.edu} (Bruce M. Boghosian)} }
\author{Bruce M. Boghosian}

\affil{Department of Mathematics, Tufts University}
\date{}

\begin{document}
\maketitle

\begin{abstract}
In this work, we modify the affine wealth model of wealth distributions to examine the effects of nonconstant redistribution on the very wealthy.  Previous studies of this model, restricted to flat redistribution schemes, have demonstrated the presence of a phase transition to a partially wealth-condensed state, or ``partial oligarchy,'' at the critical value of an order parameter. These studies have also indicated the presence of an exponential tail in wealth distribution precisely at criticality. Away from criticality, the tail was observed to be Gaussian. In this work, we generalize the flat redistribution within the affine wealth model to allow for an essentially arbitrary redistribution policy. We show that the exponential tail observed near criticality in prior work is, in fact, a special case of a much broader class of critical, slower-than-Gaussian decays that depend sensitively on the corresponding asymptotic behavior of the progressive redistribution model used.  We thereby demonstrate that the functional form of the tail of the wealth distribution in a near-critical society is not universal in nature but rather entirely determined by the specifics of public policy decisions. This is significant because most major economies today are observed to be near-critical.
\end{abstract}

\indent Wealth Distributions.

\vspace{0.05in}

\textbf{ABS Subject Classifications:}     35Q84, 35Q91, 91B80

\section{Introduction}

\subsection{Motivation}

The search for a universal form for the distribution of wealth dates back over a century to the pioneering work of Vilfredo Pareto, who first posited that wealth distribution tails are decaying power laws~\cite{Pareto2, Pareto}. While this problem may seem like a simple matter of data fitting, modern work on the subject has become vastly more complicated for at least two key reasons.

The first problem is that it is no longer sufficient to fit wealth distribution tails to particular functional forms without some microscopic model to explain the origin of those forms.  Studies over the past two decades have focused on the construction of simple models of binary transactions that can account for the form of empirical wealth distributions~\cite{Angle, Chak, AWM}.  Relating those to agent density functions and other macroscopic observable quantities then requires advanced techniques of probability theory and statistical physics~\cite{PhysicaA}.

The second problem has to do with the paucity of wealth data.  Only about a sixth of the world's countries collect reliable wealth data on their household surveys. Moreover, studies of the asymptotic behavior of the tail of the wealth distribution are necessarily focused on obfuscated data due to a small minority of households reporting their wealth.  For example, to protect anonymity, the U.S. Survey of Consumer Finance does not list the wealth of any household holding more than \$100 million, and many other countries have followed suit~\cite{SCF}.

Thus, the problem of determining how the tail of a distribution of wealth behaves is still very much an open question. In this work, we demonstrate that attempts to isolate the tail of the wealth distribution for study are, by nature, problematic.  This is, in part, because the transport equations governing wealth distribution are integrodifferential -- and hence nonlocal -- in nature~\cite{PhysicaA, AWM}.  What is happening on the tail both determines and is determined by what is happening in bulk.  Moreover, the assumption that there exists a universal form for the tail of wealth distributions is valid only when those distributions are far from a certain critical point marking the onset of wealth condensation.  Closer to that critical point, the form of the tail is subject to the minutiae of national redistribution policies rather than to any universal law of wealth distribution.  Because many of the largest economies in the world today lie near this critical point, it follows that one should not expect their wealth distribution tails to have a universal form.

There are lasting implications of this work related to how wealth distributions are viewed as economic objects. This work implies that one cannot compare the distributions of any two societies without considering the policy decisions of months, years, and decades prior.  This work emphasizes the extreme importance of redistribution in determining the form of the tail of the wealth distribution.  Our work shows that even minor policy changes can be extremely influential in large-wealth asymptotics and that a better approach to modeling wealth distributions would be to consider redistribution policy as the key driving entity determining the form of the wealth distribution tail.

\subsection{Review of literature}

The model that we use in this work is an example of an \emph{ asset exchange model}, first introduced in the
1980s~\cite{Angle} and first analyzed using methods of statistical physics in the 1990s~\cite{yakovenko, ispolatov}.  These models posit simple binary, stochastic transactions between randomly chosen pairs of agents. Our model is best understood as the result of a historical sequence of such models leading up to it.

The yard sale model -- proposed by Chakraborti in 2002~\cite{Chak} -- is an asset exchange model in which the transferred wealth is proportional to the wealth of the poorer agent in a pairwise transaction.  The small, positive proportionality captures the plausible fact that agents tend not to stake a large fraction of their total wealth in a single transaction, which models a kind of risk aversion.  Remarkably, even when the winner of a transaction is chosen with even odds in this model, wealth accumulates in the possession of a single agent, whom we call the \emph{oligarch}.  This phenomenon of a finite fraction of societal wealth belonging to a vanishingly small fraction of agents was called \emph{wealth condensation} by Bouchaud and M\'{e}zard in 2000~\cite{Bouchaud} and subsequently studied further by Burda et al.~\cite{Burda}.  Chakraborti's result may seem counterintuitive because one would expect that a system relying on a fair coin to determine the winning agent should not confer an advantage to any one economic agent.  From an economic perspective, this result is very Keynesian in that it suggests that market forces are unstable at their core and require some level of exogenous redistribution to provide stability.

In 2014, a Boltzmann equation was derived for the general yard sale model~\cite{Bogh1}.  It was also shown that this equation reduces to a nonlinear integrodifferential Fokker-Planck equation, similar to the sort used in plasma kinetic theory, where the weak-transaction limit is analogous to the weak-collision limit~\cite{Rosen}. Later in 2014, this same universal Fokker-Planck equation was shown to be derivable by means of a stochastic process~\cite{Bogh2}. The yard sale model was then extended to include a flat redistribution scheme wherein every economic agent pays an amount proportional to his or her wealth and receives a benefit proportional to the average wealth in the economy~\footnote{Equivalently stated, each economic agent is moved a certain fraction of the way toward the mean.  Those below the mean move upward, while those above the mean move downward.  The process pays for itself, as global wealth is conserved.}~\cite{Bogh2, PhysicaA}. This work showed that the oligarchical time-asymptotic state described by Chakraborti is completely mitigated under even as simple a redistribution scheme as a positive flat one~\cite{PhysicaA}. The redistribution scheme used is loosely related to the  Ornstein–Uhlenbeck process~\cite{uhlenbeck1930theory} and is similar to models of wealth redistribution introduced by others~\cite{dragulescu2000statistical,during2008kinetic,kayser2020kinetic,kayser2017kinetic, lima2020nonlinear, toscani2009wealth, toscani2010wealth}. 

The concept of wealth-attained advantage (WAA), which replaces the fair coin used in determining the winner and loser of an interaction with one that favors the richer economic agent, was also introduced to the yard sale model in 2017~\cite{PhysicaA}. The yard sale model with both WAA and flat redistribution is called the \emph{extended yard sale model} (\emph{EYSM})~\cite{PhysicaA}.  WAA is similar in spirit to the weighted coin introduced by Moukarzel et al.\cite{moukarzel2007wealth}, except that in the EYSM, the advantage conferred to the richer agent is proportional to the difference in wealth between transacting agents rather than a constant advantage conferred to the poorer agent. This extension is perhaps best motivated by James Baldwin's memorable aphorism, ``Anyone who has ever struggled with poverty knows how extremely expensive it is to be poor''~\cite{Baldwin}. Manifestations of Baldwin's precept are pervasive, and one example is the substantial difference in mortgage interest rates seen by rich and poor economic agents. The work that introduced the EYSM also introduced the concept of \emph{wealth condensation criticality}:  a state where the redistribution parameter is equal to the WAA parameter.  It was shown that the oligarchical share of wealth depends sensitively on these two parameters, as will be discussed in more detail later in this work.  For now, suffice it to say that supercritical values of the WAA parameter -- that is, values above the critical value -- will result in a partial oligarchy; subcritical values of this parameter will result in no partial oligarchy at all.

Both the yard sale model and its extensions assume that agent density has support contained in the positive real numbers so that negative wealth is not possible by construction.  In 2016, however, 10.9 \% of households in the U.S. were estimated to have negative wealth (their liabilities outweighed their assets)~\cite{SCF}. Therefore, the addition of negative wealth was seen as an important generalization that needed to be made to the EYSM~\cite{AWM, silva2004temporal}.  In 2019, the affine wealth model (AWM) was introduced~\cite{AWM}.  The AWM assumes that there is some fixed, maximum amount of debt in an economy and modifies the EYSM accordingly.  When the AWM was fit to the Survey of Consumer Finances data on wealth distributions with the Forbes 400 between 1989 and 2016, it was highly successful at modeling the U.S. wealth distribution with an average pointwise error of less than or equal to 0.16\% for each fitting.  In this work, we will extend the flat redistribution that was assumed in this work in order to examine its effects on the phenomenology of the AWM at large wealth. 

\subsection{New results}

The primary goal of this study is to consider the implications of a more general, possibly nonconstant redistribution scheme to the distribution of large wealth in the context of the above-mentioned models.  We will begin by examining the properties of the asymptotic solution to the steady-state EYSM's Fokker-Planck equation under general redistribution and will discuss the ramifications of its behavior at large wealth.  We will then generalize the asymptotic solution obtained from the EYSM to include the possibility of negative wealth and thereby obtain results for the AWM.

As mentioned above, prior work has demonstrated that when asymptotically finite redistribution functions tend toward the WAA parameter, the oligarchical share of wealth exhibits a phase transition~\cite{Bruce_SIAM_19}.  In this work, we show that exactly how the redistribution function approaches the WAA parameter in the limit of large wealth determines the nature of the tail of the distribution of wealth in a near-critical society.  In particular, small alterations in redistribution policy can radically change the nature of a distribution of wealth near criticality.  In statistical physics, it is well known that the asymptotic behavior of distributions is extremely sensitive to model parameters near criticality, and this work provides the equivalent observation for wealth distribution models.

Additionally, by solving the ``inverse problem'' of fitting data obtained from the European Central Bank (ECB) to the AWM with constant redistribution parameters, we demonstrate that all of the fourteen ECB countries that we analyzed are near-critical~\cite{ecb_data}.  This implies that the tails of these countries' wealth distributions are prone to an extreme sensitivity to the policy decisions made in that society.  This strongly suggests that minute details of those countries' redistribution policies are more important to the shape of their wealth distributions than any universal economic principles. 

\subsection{Structure of this work}

In Section \ref{sec: EYSM}, we review notation and derive the nonlinear, integrodifferential Fokker-Planck equation for the EYSM with general redistribution. We then solve for the functional form of the tail of the wealth distribution of the EYSM with general redistribution. 

In Section \ref{sec: AWM}, we generalize the asymptotic form for the tail of the solution of the EYSM with general redistribution to include the possibility of negative wealth.  We then show that the phase transition that occurs in the EYSM with general redistribution also occurs in the AWM with general redistribution, and we review the derivation of the oligarchical fraction of wealth. 

In Section \ref{sec: Inverse Problem}, we investigate the inverse problem associated with the models in Section \ref{sec: AWM}.  We assume knowledge of the large-wealth, steady-state agent density function and solve for the redistribution function corresponding to that distribution of wealth. We then examine possible functional forms for the decay of the tail of the wealth distribution in the case of progressive redistribution.  We observe a sensitive dependence of the distribution of large wealth to the form of the redistribution function in near-critical economies. Finally, we  fit the flat redistribution AWM to ECB wealth data for fourteen countries and demonstrate that all are near-critical.

\section{General redistribution in the EYSM.}
\label{sec: EYSM}

\subsection{Notation and the steady-state Fokker-Planck equation}\label{sec: SteadyState FP}

In this section, we will generalize the EYSM by allowing for more general schemes of redistribution. The main goal will be to derive a Fokker-Planck equation~\cite{ fokker1914mittlere,kolmogoroff1931analytischen, planck1917satz, Risken} for the EYSM with general redistribution~\cite{PhysicaA}.  Briefly, a Fokker-Planck equation is a partial differential equation that describes time evolution of a distribution influenced by drag and random forces.  The form for the general Fokker-Planck equation for a distribution $P(w,t)$ is given by $ \frac{\partial P}{\partial t} = -\frac{\partial}{\partial w} [\sigma(w,t) P(w,t)] +\frac{1}{2} \frac{\partial^2}{\partial w^2} [D(w,t) P(w,t)]$, where $\sigma(w,t)$ and $D(w,t)$ are called the drift and diffusion coefficients, respectively. 

We  can  describe  a  wealth  distribution  in  the  context  of  the EYSM through the use of the agent density function, $P(w,t)$. In this section, we will assume all agents have nonnegative wealth, so that $P(w,t)$ has support contained by $[0,\infty)$~\cite{PhysicaA, AWM, silva2004temporal}. We will relax this assumption in Section \ref{sec: AWM}. We define the agent density function $P(w,t)$ to be a distribution such that $\int_a^b dw\; P(w,t)$ describes the number of economic agents with wealth between $a$ and $b$ at time $t$.  It follows that $\int_a^b dw\; P(w,t)w$ describes the total wealth of those agents.  Hence, the total population and total wealth can be derived from the agent density function: $N = \int_0^\infty dw\;P(w,t)$ and $ W= \int_0^\infty dw\;P(w,t)w$, respectively. In general, total population and wealth are conserved by variations of the yard sale model~\cite{PhysicaA,Chak}, which is why we do not attach time dependence to $N$ and $W$. While it may be of concern that total wealth is constant as a function of time, prior work has shown that inflating wealth in the yard sale model (either additively or multiplicatively) does not affect the shape of the agent density function in the limit of large time~\cite{Bogh1}. We let $\mu = W/N$ be the average wealth. 

We now will introduce the Pareto-Lorenz potentials, which will be of great use in our later derivations~\cite{lorenz1905methods,Pareto}. Vilfredo Pareto and Max Lorenz were among the first academics to try to characterize the distributional nature of wealth, and their contributions are still in use in the study of wealth distributions. Mathematically, the Pareto-Lorenz potentials are the zeroth through second incomplete moments of the agent density function:
\begin{align}
    A(w,t) :&= \frac{1}{N}\int_w^\infty dx\;P(x,t) \label{eq:Aw},\\
    L(w,t) :&= \frac{1}{W}\int_0^w dx\;P(x,t)x\label{eq:Lw},\\
    B(w,t) :&= \frac{1}{N}\int_0^w dx\; P(x,t)\frac{x^2}{2}.\label{eq:Bw}
\end{align}
The function $A(w,t)$ describes the fraction of agents with wealth greater than or equal to $w$ at time $t$. Similarly, $L(w,t)$ describes the fraction of wealth held by agents whose wealth is less than or equal to $w$ at time $t$. There is no easy economic interpretation for $B(w,t)$, but it will nonetheless be useful in our derivations. Because this section assumes nonnegative wealth, $A(w,t)$ uniformly decreases as wealth increases, while $L(w,t)$ and $B(w,t)$ uniformly increase with $w$. So long as the support of $P(w,t)$ is a subset of the nonnegative real numbers, the range of $A(w,t)$ and $L(w,t)$ is $[0,1]$. Notably, because $A$ and $L$ are ratios of agents and wealth, respectively, shifts in total population or total wealth are unlikely to affect the values these functions take.

In the EYSM, each economic agent gives a fraction of their wealth to a pool that is redistributed equally across the wealth spectrum. The rate paid for redistribution is assumed to be constant across the wealth spectrum in this model~\cite{PhysicaA}. However, wealth redistribution is typically nonconstant and usually progressive~\cite{lindert2017rise}. To generalize redistribution in the EYSM, we will now introduce the concept of a redistribution function $\chi(w)$, which is at the moment arbitrary apart from the assumption that $\chi(w)P(w,t)w$ is globally integrable for all $t$. The redistribution function $\chi(w)$ will be a function of wealth that returns the rate of redistribution paid by an agent with wealth $w$. It is easy to see that the total wealth collected for redistribution at time $t$ is given by $T(t)~=~\int_0^\infty dx\; P(x,t)\chi(x)x$.  At time $t$, we assume that an economic agent receives a benefit proportional to the average redistribution collected: $\frac{T(t)}{N}$. Note that at any given time, all redistribution is distributed among agents so that wealth is conserved. That is, redistribution moves from agents to agents and not from agents to a government body. Thus, redistribution is considered distinct from taxation in our model.

We can construct the Fokker-Planck equation for the EYSM with general redistribution by considering a transaction between an agent with wealth $\bar{w}$ and another agent with wealth $\bar{x}$ at time $t$. If a transaction takes place in time $\Delta t$, the agent with wealth $\bar{w}$ faces a change in wealth of
\begin{align}
  \Delta w = \bigg[\frac{T(t)}{N}-\chi(\bar{w})\bar{w}\bigg]\Delta t + \sqrt{\gamma \Delta t}\min\{\bar{w},\bar{x}\}\eta,\label{eq:GeneralRedistributionMicrotransaction}
\end{align}
where $\gamma$ describes the number of transactions per-unit-time and $\eta\in\{-1,1\}$ is a stochastic random variable to be discussed in more detail shortly. The value $\sqrt{\gamma \Delta t}$ provides the proportionality and time dependence needed in a given transaction, as will become clear soon.

To include the concept of WAA into this model~\cite{PhysicaA}, we shift the expectation of $\eta$:
\begin{align}
  E[\eta] &= \zeta \sqrt{\frac{\Delta t}{\gamma}}\;\frac{\bar{w}-\bar{x}}{\mu}.\label{eq:WAADefinition}
\end{align}
It follows from $\eta\in\{-1,1\}$ that $E[\eta^2] = 1$. Note that, in  (\ref{eq:WAADefinition}), the expectation of $\eta$ is proportional to the difference in wealth between the two transacting agents and that we have introduced the proportionality constant $\zeta>0$, which is a per-unit-time constant that quantifies the level of WAA in a society. Under this construction, a larger $\zeta$ corresponds to a society in which wealthier agents are given a greater advantage in transactions, while agents of different wealth transact on a more level playing field when $\zeta$ is close to 0. 

Using (\ref{eq:GeneralRedistributionMicrotransaction})-(\ref{eq:WAADefinition}), we can derive the Fokker-Planck equation for the EYSM with general redistribution. Following earlier work~\cite{Bogh2,PhysicaA}, we use the following notation for the bivariate expectation of a function of both $\eta$ and $\bar{x}$ at time $t$:
\begin{align*}
    \mathcal{E}[f(\eta,\bar{x})] = \int_0^\infty d\bar{x}\; P(\bar{x},t) E[f(\eta,\bar{x})],
\end{align*}
where $E[f(\eta,\bar{x})]$ denotes the expectation of $f(\eta,\bar{x})$ with respect to $\eta$. Using this notation, we are now able to derive the drift and diffusion coefficients of the Fokker-Planck equation for the EYSM with general redistribution,
\begin{align*}
\begin{split}
\sigma(w,t) = \lim_{\Delta t\rightarrow 0} \mathcal{E}\bigg[\frac{\Delta w}{\Delta t}\bigg]  &= \bigg(\frac{T}{N}-\chi(w)w \bigg)- \zeta\bigg[\frac{2}{\mu}\bigg(B-\frac{w^2}{2}A\bigg) + (1-2L)w\bigg]
\end{split},\\
     D(w,t) = \lim_{\Delta t\rightarrow 0} \mathcal{E}\bigg[\frac{(\Delta w)^2}{\Delta t}\bigg]&= 2\gamma \bigg(B+ \frac{w^2}{2}A\bigg),
\end{align*}
where $A$, $L$, and $B$ are the Pareto-Lorenz potentials defined in (\ref{eq:Aw})-(\ref{eq:Bw}). For ease of notation, we drop functional dependence, but at this point, $P$, $A$, $L$, and $B$ are to be understood as functions of both wealth and time, and $T$ is to be understood as a function of just time. 

We thus arrive at the Fokker-Planck equation of the EYSM~\cite{PhysicaA}, generalized to allow for arbitrary redistribution schemes:
\begin{align}
\begin{split}
\frac{\partial P}{\partial t} =& -\frac{\partial}{\partial w}\bigg[\bigg(\frac{T}{N}-\chi(w)w\bigg)P\bigg] + \frac{\partial}{\partial w}\bigg\{\zeta\bigg[ \frac{2}{\mu}\bigg(B-\frac{w^2}{2}A\bigg) + (1-2L)w\bigg]P\bigg\}\\
 &\quad +  \frac{\partial^2}{\partial w^2}\bigg[\gamma \bigg(B+\frac{w^2}{2}A\bigg)P\bigg].
\end{split}\label{eq:FokkerPlanckwithGamma}
\end{align}
Because $\gamma$ is independent of wealth and time, we can divide  (\ref{eq:FokkerPlanckwithGamma}) by $\gamma$ and absorb it into the timescale by sending $\gamma t$ to $t$~\cite{PhysicaA}. Furthermore, we let $\bar{\chi}(w)=\frac{1}{\gamma}\chi(w)$ and $\bar{\zeta}=\frac{1}{\gamma}\zeta$. These simplifications change the units of redistribution and WAA to a per-transaction basis instead of per-unit-time. Henceforth, we will forgo the barring of $\chi(w)$ and $\zeta$ for the sake of notation, but both of these quantities are to be understood as being on a per-transaction basis. Using these simplifications, we arrive at the following nonlinear integrodifferential equation for agent density:
\begin{align}
\begin{split}
\frac{\partial P}{\partial t} =& -\frac{\partial}{\partial w}\bigg[\bigg(\frac{T}{N}-\chi(w)w\bigg)P\bigg] + \frac{\partial}{\partial w}\bigg\{\zeta\bigg[ \frac{2}{\mu}\bigg(B-\frac{w^2}{2}A\bigg) + (1-2L)w\bigg]P\bigg\}\\
 &\quad + \frac{\partial^2}{\partial w^2}\bigg[ \bigg(B+\frac{w^2}{2}A\bigg)P\bigg].
\end{split}\label{eq:FokkerPlanck}
\end{align}
We set the time derivative of  (\ref{eq:FokkerPlanck}) to zero and integrate once with respect to wealth to obtain the following first-order nonlinear nonlocal ordinary differential equation describing the steady-state wealth distribution:
\begin{align}
    \frac{d}{dw}\bigg[\bigg(B+\frac{w^2}{2}A\bigg) P \bigg] = \bigg\{\frac{T}{N}-\chi(w)w - \zeta\bigg[\frac{2}{\mu}\bigg(B-\frac{w^2}{2}A\bigg) + (1-2L)w\bigg]\bigg\}P. \label{eq:FokkerPlanckSteadyState}
\end{align}
Note that in  (\ref{eq:FokkerPlanckSteadyState}), we have dropped dependence on time and have used total derivatives with respect to wealth. We will henceforth follow this precedent because we will only be dealing with steady-state wealth distributions for the duration of this work. While agent density is likely to change as time evolves, recent work has showed that, in a reasonably stable economy, the time derivative of $P$ can be treated  as a higher-order perturbation of the steady-state wealth distribution~\cite{sheng}. 

\subsection{Large wealth analysis of the steady-state Fokker-Planck equation}\label{sec: EYSM large w-2.2}

\subsubsection{Assumptions on agent density at large wealth}\label{sec: Assumptions}

We now establish notation that will be used extensively throughout this work. We will use the notation $g(w)\ll h(w)$ to mean
\begin{align}
    \lim_{w\rightarrow \infty }\frac{g(w)}{h(w)} = 0.\label{eq:AsymptoticDominance}
\end{align}
An equivalent notation often used for this functional relationship is $g(w) = o[h(w)]$. Furthermore, we will say that $g(w) \approx h(w)$ if 
\begin{align}
    \lim_{w\rightarrow \infty }\frac{g(w)}{h(w)} = 1.\label{eq:AsymptoticEquality}
\end{align}
We now will provide a set of assumptions that will enable us to approximate the large-wealth behavior of agent density. Using the above notation, we make the a priori assumption that the redistribution function $\chi(w)$ satisfies
\begin{align*}
   \frac{d\chi}{dw}\ll w[\chi(w) + \alpha w+\beta]^2
\end{align*}
for any real constants $\alpha$ and $\beta$. Note that this condition is very general and that functional forms for $\chi(w)$ ranging from arbitrary polynomials to the exponential of arbitrary polynomials will satisfy it. 

We now make some a posteriori assumptions on the distribution of large wealth. Assume that 
\begin{align}
    P(w) \approx C e^{-f(w)} + cW \Xi(w),\label{eq: P_Definition}
\end{align}
where $f(w)$ is a twice-differentiable, asymptotically monotone function defined on $(M,\infty)$ -- for some $M>0$ -- satisfying the following conditions:
\begin{align}
    f'(w) &>0, \label{eq:Assumption1}\\
    f'(w) &\gg \sqrt{f''(w)},\label{eq:Assumption2}\\
    e^{f(w)} &\gg \frac{w^2}{[f'(w)]^2}. \label{eq:Assumption3}
\end{align}
We assume that $C$ is a positive constant of integration~\cite{PhysicaA}. We will use only redistribution functions $\chi(w)$ for which (\ref{eq:Assumption1})-(\ref{eq:Assumption3}) are valid. These constraints are very lax and are satisfied, inter alia, by all functions of the form  $f(w) \approx  w^p\log^q(w),$ where either $p>0$ and $q\in \mathbb{R}$ or $p=0$ and $q>1$. This is shown in Appendix \ref{sec: satisfactory functions}.

The function $\Xi(w)$ given in  (\ref{eq: P_Definition}) is a generalized distribution that was introduced in prior work to represent the oligarchical fraction of wealth in an economy~\cite{PhysicaA, Bruce_SIAM_19}. There is a broad literature on the behavior of $\Xi(w)$, which we leave to the reader to investigate further. Briefly, however, $\Xi(w)$ satisfies
\begin{align*}
    \int_0^\infty dw\; \Xi(w) &= 0,\\
    \int_0^\infty dw\; \Xi(w)w &= 1,\\
    \int_a^b dw\; \Xi(w)w &= 0
\end{align*}
for any $a,b\in\mathbb{R}$. The constant $c$ is a number between 0 and 1 representing the fraction of wealth in the possession of an oligarch. This constant will be analyzed in further detail in Section \ref{sec: EYSM Agent Density 2.2.2}.

\subsubsection{Agent density at large wealth}\label{sec: EYSM Agent Density 2.2.2}

In this section, we will derive an approximate form for $f(w)$ based on the assumptions given in (\ref{eq:Assumption1})-(\ref{eq:Assumption3}).  Prior work leads us to posit that it is the behavior of redistribution at large wealth which enables or inhibits oligarchy~\cite{PhysicaA,Bogh6, Bruce_SIAM_19, silva2004temporal}. By better understanding the functional form of agent density for the wealthiest agents in the distribution, we hope to better understand which redistribution policies are sufficient to preclude oligarchy. The following lemma is necessary to make the approximations which serve as the foundation for this work. It can easily be proven through  (\ref{eq:AsymptoticEquality}), an application of L'H\^{o}pital's rule, and the consequences of our assumptions, which are provided in Appendix \ref{sec: consequences of assumptions}. 
\begin{lemma}\label{lemma:ParetoPotentialApproximation}
Under the assumptions listed in (\ref{eq:Assumption1})-(\ref{eq:Assumption3}), when \newline wealth is sufficiently large and $m\geq 0$,
\begin{align*}
    \int_w^\infty dx\; x^m \exp[-f(w)] &\approx \frac{w^m}{f'(w)}\exp[-f(w)].
\end{align*}
\end{lemma}

Note that we can define $L(w)$ and $B(w)$ in terms of integrals that are considered by Lemma \ref{lemma:ParetoPotentialApproximation}. In particular, at equilibrium,
\begin{align*}
    L(w) &= L_\infty-\frac{1}{W}\int_w^\infty dx\; P(x)x,\\
    B(w) &= B_\infty - \frac{1}{N}\int_w^\infty dx\;  P(x)\frac{x^2}{2},
\end{align*}
where $L_\infty$ and $B_\infty$ are the complete first and second moments of $P(w)$ with respect to $w$. Prior work has shown these to be finite numbers~\cite{PhysicaA}. Intuitively, one would expect $L_\infty$ to be 1 because $W$ is defined to be the first complete moment of $P$ with respect to wealth. In fact, this reasoning holds even for the asymptotic behavior of $A(w)$, and it can be proven that $A(w)$ tends toward zero as $w\rightarrow\infty$. However, numerical simulations and analytic studies have shown that this orderly convergence does not hold for $L(w)$~\cite{ PhysicaA, Bruce_SIAM_19}. It has been shown that, for asymptotically constant redistribution functions $\chi(w) \approx \chi_\infty >0$,
\begin{align}
    L_\infty :&= \lim_{w\rightarrow \infty} L(w)  = \begin{cases} 1 & \text{ if }\zeta\leq\chi_\infty,\\
    \frac{\chi_\infty}{\zeta} &\text{ if }\zeta>\chi_\infty
    \end{cases}\label{eq:L_inf_EYSM_specific}
\end{align}
\cite{PhysicaA, Bruce_SIAM_19}. 
This phenomenon was shown to be due to a second-order phase transition observed at criticality -- a state defined by $\chi_\infty = \zeta$ -- that is due to $\Xi(w)$~\cite{PhysicaA, Bruce_SIAM_19}. Prior work observed that if $\zeta>\chi_\infty$ -- a state called \emph{supercritical} -- there is a partial oligarch with fraction of total wealth $c=1-\frac{\chi_\infty}{\zeta}$. Conversely, if $\zeta<\chi_\infty$ -- a state called \emph{subcritical} -- there is no partial oligarchy whatsoever ($c=0$). This relationship between redistribution and WAA explains the duality exhibited in $L_\infty$ in  (\ref{eq:L_inf_EYSM_specific}). In both subcritical and supercritical distributions, the distribution of large wealth is observed to be Gaussian. However, if $\chi_\infty=\zeta$ -- a state called \emph{critical} -- the oligarchical fraction of wealth drops to zero, and the distribution of large wealth is observed to be exponential~\cite{PhysicaA, Bruce_SIAM_19}. 

Reducing the steady-state Fokker-Planck equation requires significant algebra, which is provided in Appendix \ref{sec: Approximations}. The final result is that  (\ref{eq:FokkerPlanckSteadyState}) reduces to
\begin{align}
    f(w) &\approx \frac{1}{B_\infty}\int^w dx\; \chi(x)x + \frac{\zeta(1 - 2L_\infty)}{2B_\infty}w^2 +\frac{2\zeta B_\infty  - \frac{T}{N}\mu }{B_\infty \mu}w.\label{eq:LargeWealthLogPw}
\end{align}
The lower limit of integration in the redistributive term of  (\ref{eq:LargeWealthLogPw}) is omitted because it is a subdominant constant of integration. Note that for any given redistribution function $\chi(w)$ and WAA parameter $\zeta$ satisfying our assumptions, we can use   (\ref{eq:LargeWealthLogPw}) to find the distribution of large wealth. We emphasize that this form of $f(w)$ is both an extension and a corroboration of prior research on the EYSM at large wealth~\cite{PhysicaA}. In that work, the steady-state Fokker-Planck equation's asymptotic solution was found to be
\begin{align}
    f(w)\approx \frac{|\chi-\zeta|}{2B_\infty}w^2 +  \frac{2\zeta B_\infty  - \chi\mu^2 }{B_\infty \mu}w, \label{eq:LargeWealthLogPw_ConstantChi}
\end{align}
which is derivable from (\ref{eq:L_inf_EYSM_specific}) and (\ref{eq:LargeWealthLogPw}) if redistribution is assumed to be constant~\cite{PhysicaA}.


\section{General redistribution in the AWM}\label{sec: AWM}

In Section \ref{sec: EYSM}, we solved for the large-wealth distribution of agent density under the assumption that wealth is nonnegative. However, economic agents with negative wealth are widely observed in real-world data. For example, in 2016, 10.9 \% of the population of the U.S. was estimated to have negative wealth~\cite{bricker2017changes,SCF}. In this section, we extend our earlier generalization of redistribution in the EYSM to a generalization of redistribution in the AWM. We will use the work of Section \ref{sec: EYSM} as the basis for this extension. 

The AWM is a recently-introduced asset exchange model that allows for negative wealth. The AWM has been highly successful at modeling empirical wealth data~\cite{AWM}. It is based on the EYSM but allows the support of the agent density function to be contained in $[-\Delta,\infty)$, where $\Delta\geq 0$ is the fixed maximal value of debt in an economy. Thus, economic agents can forfeit more wealth than their net worth, but only up to a set limit. The AWM was constructed so that transacting agents have the same economic relationship to each other as if their wealth was shifted upward by $\Delta$. In particular, before a given transaction, economic agents add $\Delta$ to the wealth. They will therefore have positive wealth for the duration of the transaction, and the results of the EYSM and Section \ref{sec: EYSM} will apply. At the end of the transaction, $\Delta$ is subtracted from both agents' wealth. Thus, the AWM shifts the agent density function of the EYSM on the wealth axis by a factor of $\Delta$. 

In this section, we will follow prior notation and bar quantities that are used within the context of the EYSM~\cite{AWM}. We will let the unbarred quantities refer to their AWM equivalents. For example, $\bar{P}(w)$ will denote the steady-state agent density function explored in Section \ref{sec: EYSM}, and  $P(w)$ will denote the steady-state ``shifted wealth'' agent density function of the AWM. We assume that $P(w)$ has support contained within $[-\Delta,\infty)$. The following algebraic manipulation can be easily observed, linking $P(w)$ and $\bar{P}(w)$:
\begin{align}
    P(w) = \bar{P}(w+\Delta). \label{eq:EYSM_to_AWM_Pw}
\end{align}
If we make the same assumptions on agent density as in Section \ref{sec: Assumptions} -- given in (\ref{eq:Assumption1})-(\ref{eq:Assumption3}) -- we arrive at an approximate form for $\bar{f}(w) \approx -\log[\bar{P}(w)] $, given by  (\ref{eq:LargeWealthLogPw}). By (\ref{eq:LargeWealthLogPw}) and (\ref{eq:EYSM_to_AWM_Pw}), the distribution of large wealth is given by
\begin{align*}
    f(w) \approx& -\log[P(w)]-\log(C)\nonumber\\
    \approx& \bar{f}(w+\Delta)-\log(C) \nonumber\\
    \begin{split}
    \approx& \frac{1}{\bar{B}_\infty}\int^{w+\Delta} dx\; \chi(x)x + \frac{\zeta(1 - 2\bar{L}_\infty)}{2\bar{B}_\infty}(w+\Delta)^2 +\frac{2\zeta \bar{B}_\infty  - \frac{T}{N}\bar{\mu} }{\bar{B}_\infty \bar{\mu}}(w+\Delta)\\&\quad-\log(C) \nonumber.
    \end{split}
\end{align*}
Next, we will expand the quadratic term and group by the power in wealth,
\begin{align}
    \begin{split}
        =& \frac{1}{\bar{B}_\infty}\int^{w+\Delta} dx\; \chi(x)x + \frac{\zeta(1 - 2\bar{L}_\infty)}{2\bar{B}_\infty}w^2 +\frac{2\zeta \bar{B}_\infty  - \frac{T}{N}\bar{\mu} + \zeta \Delta \bar{\mu}(1 - 2\bar{L}_\infty) }{\bar{B}_\infty \bar{\mu}}w\\
        &\quad-\log(C), \label{eq: LargeWealthLogPw_AWM}
    \end{split}\\
    =& \bar{f}(w) +\frac{1}{\bar{B}_\infty}\int_w^{w+\Delta} dx\; \chi(x)x +  \frac{\zeta\Delta (1-2\bar{L}_\infty)}{\bar{B}_\infty}w, \label{eq: LogPw_correction_terms}
\end{align}
where, without loss of generality, the constant term has been absorbed into the constant of integration $C$. Hence, for any redistribution function $\chi(w)$ and WAA parameter $\zeta$, there exists a function $f(w)$ such that the agent density decays according to $P(w) \approx C e^{-f(w)} + c W \Xi(w)$. In  (\ref{eq: LogPw_correction_terms}), we provide the AWM distribution at large wealth in terms of the EYSM distribution with correction terms.

In Section \ref{sec: EYSM Agent Density 2.2.2}, we reviewed the results of prior work that showed the existence of a second-order phase transition in the EYSM. This phase transition in the oligarchical share of wealth occurred whenever the limit of the  redistribution $\chi_\infty$ is less than $\zeta$~\cite{PhysicaA,Bogh6, Bruce_SIAM_19}. An analogous result holds for the AWM. We will devote the rest of this section reviewing the derivation of the oligarchical share of wealth in the AWM. For ease of notation, we introduce the parameter $\lambda\geq0$, which is defined implicitly by $\Delta = \lambda\bar{\mu}$. Thus, a larger $\lambda$ corresponds to a larger maximum value of debt in the model. 

Let $\bar{L}(w)$ be the Pareto-Lorenz potential defined in  (\ref{eq:Lw}) for the EYSM, and let $L(w)$ be its AWM equivalent. These two integral operators can be shown to be related in the following way:
\begin{align*}
    L(w) = (1+\lambda )\bar{L}(w+\Delta)-\lambda [1-\bar{A}(w+\Delta)],
\end{align*}
where $\bar{A}(w)$ is the EYSM Pareto-Lorenz potential defined in  (\ref{eq:Aw})~\cite{AWM}. Note that $\bar{A}(w+\Delta)$ and $\bar{L}(w+\Delta)$ are Pareto potentials in the EYSM. Therefore, as $w\rightarrow\infty$, $\bar{A}(w+\Delta)\rightarrow 0$ and $\bar{L}(w+\Delta)\rightarrow \bar{L}_\infty$: the asymptotic limit of $\bar{L}(w)$ described in  (\ref{eq:L_inf_EYSM_specific}). Hence,
\begin{align}
   L_\infty = \lim_{w\rightarrow\infty} \bigg\{(1+\lambda )\bar{L}(w+\Delta)-\lambda [1-\bar{A}(w+\Delta)]\bigg\} = (1+\lambda )\bar{L}_\infty -\lambda.\label{eq: L_inf_AWM_general}
\end{align}

By  (\ref{eq: L_inf_AWM_general}), the fraction of wealth held by the oligarch in the AWM is closely related to that which was introduced in the EYSM. Moreover, due to the inclusion of $\bar{L}_\infty$ in  (\ref{eq: L_inf_AWM_general}), the AWM exhibits a similar phase transition to that of the EYSM. Suppose that $\chi_\infty\geq \zeta$, as would be the case in a subcritical or critical economy. By (\ref{eq:L_inf_EYSM_specific}) and (\ref{eq: L_inf_AWM_general}), 
\begin{align*}
    L_\infty = (1+\lambda)\times 1 - \lambda = 1.
\end{align*}
This implies that the oligarchical share of wealth is zero. Suppose that $\chi_\infty<\zeta$, as would be the case in an EYSM supercritical economy. By (\ref{eq:L_inf_EYSM_specific}) and (\ref{eq: L_inf_AWM_general}), 
\begin{align*}
    L_\infty = (1+\lambda) \frac{\chi_\infty}{\zeta} - \lambda.
\end{align*}
Hence, the fraction of wealth held by the oligarch is  $c= 1-L_\infty = (1+\lambda) \Big(1-\frac{\chi_\infty}{\zeta}\Big).$ In particular,  the fraction of wealth held by the partial oligarch is magnified by the amount of negative wealth in a society. These results are summarized in the following equation:
\begin{align}
    L_\infty = \begin{cases} 1 & \chi_\infty \geq \zeta,\\
    (1+\lambda) \frac{\chi_\infty}{\zeta} - \lambda & \chi_\infty<\zeta.
    \end{cases} \label{eq:L_inf_AWM_specific}
\end{align}
Note that if $\lambda = 0$, as would be the case if all wealth were nonnegative,  (\ref{eq:L_inf_AWM_specific}) reverts back to the asymptotics of $\bar{L}(w)$ in the EYSM. Hence, this argument truly does generalize the EYSM with general redistribution to include the possibility of negative wealth.


\section{Inverse problem}\label{sec: Inverse Problem}

\subsection{Redistribution as a function of large-wealth agent density}\label{sec: Redistribution Derivation}

After deriving  (\ref{eq: LargeWealthLogPw_AWM}), we observed that for any redistribution function $\chi(w)$ and WAA parameter $\zeta$, there exists a function $f(w)$ such that agent density decays according to $P(w) \approx C e^{-f(w)} + cW\Xi(w)$, where $c$ depends on the limit of $\chi(w)$ relative to $\zeta$. It is worth noting that the converse of this statement may not be true. In this section, we will delve into the inverse problem to Section \ref{sec: AWM}. In particular, given that policymakers would like the distribution of large wealth to obey $P(w)\approx C e^{-f(w)}+cW\Xi(w)$, what redistribution policy should they follow? Is the distribution they have in mind even possible within the context of the AWM? In this section, we aim to provide a formal answer to these questions. 

Throughout this section, we will assume knowledge of a function  $f(w)$ -- twice differentiable on $(M,\infty)$ for some $M>-\Delta$ -- such that agent density at large wealth is of the form $P(w)\approx C e^{-f(w)}+ cW\Xi(w)$. Thus, we allow for the possibility of oligarchy in our a priori assumption on $P(w)$. However, because we are primarily concerned with the distribution of large wealth and not that of oligarchs, we will not delve into the deep and interesting literature on $\Xi(w)$. We refer any readers interested in the oligarch's contribution to the distribution of wealth to prior work on this subject~\cite{Bogh6, Bruce_SIAM_19}. For now, the constant $c=1-L_\infty$ can be thought of as a parameter to be tuned by legislators when drafting policy. Many of the following derivations will be analogous to those in Section \ref{sec: EYSM}. The important distinction is that in Section \ref{sec: EYSM}, we assumed knowledge of $\chi(w)$ and an a priori asymptotic form $P(w)\approx C e^{-f(w)}+cW\Xi(w)$ for some function $f(w)$ that satisfies (\ref{eq:Assumption1})-(\ref{eq:Assumption3}). In this section, we assume knowledge of a $P(w)$ satisfying (\ref{eq:Assumption1})-(\ref{eq:Assumption3}) -- and hence the asymptotic form $P(w)\approx C e^{-f(w)} + cW\Xi(w)$ -- but make the a priori assumption that there is a $\chi(w)$ that will produce the asymptotic distribution of $P(w)$. 

As in Section \ref{sec: SteadyState FP}, we set the time derivative of the EYSM's Fokker-Planck equation with general redistribution --  (\ref{eq:FokkerPlanck}) -- equal to zero and integrate once with respect to wealth to obtain  (\ref{eq:FokkerPlanckSteadyState}). Note that Lemma \ref{lemma:ParetoPotentialApproximation} is still applicable for asymptotic approximations, as any EYSM agent density function that we will consider satisfies the conditions stated in (\ref{eq:Assumption1})-(\ref{eq:Assumption3}) by assumption. Then the approximations that were made in Appendix \ref{sec: Approximations} to derive the asymptotic form of $\bar{f}(w)$ from the steady-state EYSM Fokker-Planck equation are still valid, and we arrive at  (\ref{eq:LargeWealthLogPw}), as in Section \ref{sec: EYSM Agent Density 2.2.2}. We then can apply the transformation discussed in Section \ref{sec: AWM} to obtain a form for $f(w)$: the negative logarithm of the AWM agent density function. At this point, we rearrange  (\ref{eq: LargeWealthLogPw_AWM}) to find that the redistribution function $\chi(w)$ must satisfy
    \begin{align*}
   \begin{split}
       \int^{w+\Delta} dx\; \chi(x)x =& \bar{B}_\infty \bar{f}(w) + \frac{\zeta( 2\bar{L}_\infty-1)}{2}w^2 +\frac{2\zeta \bar{B}_\infty  - \frac{T}{N}\bar{\mu} + \zeta \Delta \bar{\mu}(1 - 2\bar{L}_\infty) }{\bar{\mu}}w
   \end{split}
   \end{align*}
   if it exists for a given distribution $P(w)$. Applying the fundamental theorem of calculus and dividing by $w$, we find that $\chi(w)$ must have the asymptotic form
   \begin{align}
       \chi(w+\Delta) =& \zeta(2\bar{L}_\infty - 1) + \bar{B}_\infty \frac{\bar{f}'(w)}{w} + \frac{2\zeta \bar{B}_\infty  - \frac{T}{N}\bar{\mu} + \zeta \Delta \bar{\mu}(1 - 2\bar{L}_\infty) }{\bar{\mu}}\frac{1}{w}. \label{eq:Redistribution1}
\end{align}
We can thus find $\chi(w)$ by inputting $w-\Delta$ on the right hand side of  (\ref{eq:Redistribution1}). We now will apply a Maclaurin expansion of $\chi(w)$ about $\Delta$,
\begin{align*}
    \begin{split}
        \chi(w) =& \zeta(2\bar{L}_\infty - 1) + \bar{B}_\infty \frac{\bar{f}'(w)}{w} + \frac{2\zeta \bar{B}_\infty  - \frac{T}{N}\bar{\mu} }{\bar{\mu}}\frac{1}{w}+ O\bigg(\frac{\Delta}{w^2}\bigg),
    \end{split}
\end{align*}
where we have used big-$O$ notation to refer to subdominant terms in the Taylor expansion. We assume that $w$ is very large so that $\Delta$ is small relative to $w$. Then $\Delta$ is even smaller compared to $w^2$, and we may approximate $\chi(w)$ by its leading terms,
\begin{align}
    \chi(w) \approx& \zeta(2\bar{L}_\infty - 1) + \bar{B}_\infty \frac{\bar{f}'(w)}{w} + \frac{2\zeta \bar{B}_\infty  - \frac{T}{N}\bar{\mu}  }{\bar{\mu}}\frac{1}{w}.\label{eq:Redistribution2}
\end{align}
The reader may be concerned that $\chi(w)$ is defined in terms of $T$, which is itself a functional of $\chi(w)$. We note that for all redistribution functions we consider relevant, the total redistribution will be constant. For this reason, we assume that it is an arbitrary constant for the purpose of solving  (\ref{eq:Redistribution2}) and assume that its value can be set once $\chi(w)$ and $P(w)$ are known. 

 (\ref{eq:Redistribution2}) implies that we can describe asymptotic redistribution as the sum of a constant and some function of wealth. In particular, we let 
\begin{align*}
    \chi(w) :&\approx \zeta(2\bar{L}_\infty -1) + \iota(w),
\end{align*}
where $\iota(w)$ is defined by
\begin{align*}
    \iota(w) :&\approx \bar{B}_\infty \frac{\bar{f}'(w)}{w} + \frac{2\zeta \bar{B}_\infty  - \frac{T}{N}\bar{\mu}  }{\bar{\mu}}\frac{1}{w}.
\end{align*}
We note that the behavior of $\iota(w)$ -- as the sole nonconstant contribution to redistribution at large wealth -- will dictate the large-wealth behavior of redistribution. At this point, we are free to consider the implications of the derived form of $\chi(w)$. To do this, we introduce some important terminology. Assume  $P(w)\approx C e^{-f(w)} + cW\Xi(w)$, where $f(w)$ is a differentiable function. We say that $P(w)$ decays \emph{subquadratically} if $f'(w)\ll w$. Similarly, we say that $P(w)$ decays \emph{quadratically} if $f'(w)\approx a w$ for some $a\neq 0$. Finally, we say that $P(w)$ decays \emph{superquadratically} if $f'(w)\gg w$.

\subsection{Criticality in the AWM with general redistribution} \label{sec: Criticality}

Prior work on the AWM observed a critical relationship in the case where the redistribution function $\chi(w)$ is asymptotically constant, tending toward a limit we shall call $\chi_\infty$~\cite{Bruce_SIAM_19}. When $\chi_\infty<\zeta$, agent density was observed to decay like a Gaussian, and there was a partial oligarchy with share of wealth $c=1-\frac{\chi_\infty}{\zeta}$.  By contrast, when $\chi_\infty>\zeta$, agent density decayed like a Gaussian with no oligarchy whatsoever ($c=0$). However, when $\chi_\infty=\zeta$ -- a state introduced earlier as \emph{criticality} -- agent density was observed to decay exponentially with no oligarch.  In this subsection, we will show that the critical exponential decay that was observed in prior work is actually a special case of a more general family of subquadratic decays~\cite{PhysicaA, Bruce_SIAM_19}. We will prove that there exists a family of subquadratic decays with redistribution $\chi(w) \approx \zeta$ and no oligarchy. 

Suppose that $f(w)$ is a twice-differentiable function on $(M,\infty)$ for some $M>-\Delta$, that it satisfies (\ref{eq:Assumption1})-(\ref{eq:Assumption3}), and that $f'(w)\ll w$: a condition which corresponds to agent density decaying subquadratically. It is clear from  (\ref{eq:Redistribution1}) that $\iota(w)\rightarrow 0$ as $w\rightarrow\infty$, so that $\chi(w)\rightarrow\zeta(2\bar{L}_\infty-1)$. However, the construction of $\chi(w)$ is such that all higher-order (quadratic and linear) terms are canceled in the asymptotic form for $f(w)$. So, $\iota(w)$ contributes nonnegligibly to large-wealth agent density. Prior work has dealt only with constant redistribution schemes. Therefore the critical exponential distribution -- the case in which $\chi(w) = \chi = \zeta$ and $\bar{L}_\infty = 1$ -- seemed like the unique subquadratic decay satisfying these properties~\cite{Bruce_SIAM_19}. The above argument extends this idea to a family of subquadratic distributions. We have assumed nothing about the existence or nonexistence of oligarchy in the subquadratic case, but note that if $\bar{L}_\infty = 1$, this argument describes redistribution functions $\chi(w)$, tending toward $\zeta$ as $w\rightarrow\infty$, which will produce subquadratic decays that are not exponential.

We have shown that many subquadratic decays are possible by allowing nonconstant redistribution and that to attain such a distribution of large wealth, $\chi(w)$ must tend toward $\zeta(2\bar{L}_\infty-1)$ as $w\rightarrow\infty$.  However, the way that $\chi(w)$ approaches this limit warrants further discussion. Suppose that a government is aiming for a subquadratic distribution of large wealth $C_1 e^{-g(w)}$ but that the limit of redistribution it aims for differs by a small margin from $\zeta(1-2\bar{L}_\infty)$. In particular, if $\epsilon\in\mathbb{R}$ is some small, possibly negative constant, suppose that $\chi(w)\rightarrow \zeta(2\bar{L}_\infty -1)+\epsilon$ as wealth becomes large. In this case, Section \ref{sec: Redistribution Derivation} implies that agent density will behave like $C_2 e^{-f(w)} + cW\Xi(w)$, where 
\begin{align}
    f(w) \approx  \frac{1}{B_\infty }\int^{w+\Delta} dx \;\bigg[ \epsilon\; x +   \bar{B}_\infty g'(x) \bigg] \label{eq:EpsilonArgument}
\end{align}
and $c = 1-L_\infty$. We have assumed that $g'(w)\ll w$, so  (\ref{eq:EpsilonArgument}) will at some point be well approximated by an order-$w^2$ term. However, there may be a section of the wealth distribution where the contribution from $g(x)$ competes with $\epsilon \; w^2 $ if $\epsilon$ is sufficiently small. 

Importantly, the point at which the linear term in the integrand of  (\ref{eq:EpsilonArgument}) dominates the $g'(w)$ term may be near the end of the wealth spectrum, where the discretization of agent density will make it irrelevant. By continuity, there is a point $w_\epsilon$ at which $g'(w_\epsilon) = \epsilon w_\epsilon$. At this point, the two terms in the integrand of  (\ref{eq:EpsilonArgument}) become comparable, but for wealth much lower than this, $P(w)\approx C_1 e^{-g(w)}$. Similarly, for wealth much greater than $w_\epsilon$, agent density will decay like a Gaussian, and the theory of prior work on criticality applies~\cite{Bruce_SIAM_19, AWM}. However, if it is true that $NA(w_\epsilon)<1$, where $A(w)$ is defined to be the AWM equivalent to the Pareto-Lorenz potential given in  (\ref{eq:Aw}), there will be no economic agents with wealth greater than $w_\epsilon$. This argument shows that the discretization of agent density allows the limit of redistribution to differ from the WAA parameter while still attaining a subquadratic distribution of wealth.

The redistribution function corresponds to the policy choices of a society, and this analysis shows that those choices are of the utmost importance for economies near criticality. The redistribution function has been broken up into its constant and nonconstant contributions. In the case of a subquadratic decays in agent density, the nonconstant contribution $\iota(w)$ tends toward zero as $w\rightarrow\infty$. However, the argument of this section has shown that the way in which $\iota(w)$ tends toward zero indicates the nature of the subquadratic decay which is attained. Hence, when an economy is near-critical, the qualitative nature of its distribution of large wealth is extremely sensitive to the trivialities of redistribution policy. 

\subsection{Asymptotic redistribution functions for common tails}\label{sec: Redistribution Functions}

The work of Section \ref{sec: Criticality} implies that it is possible to attain any subquadratic tail given some redistribution function. In the literature on wealth distribution tails, many forms have been fitted to empirical data. In this section, we will provide the asymptotic redistribution function necessary to attain common distributions in the study of tails of wealth distributions. To do this, we first convert a probability density function to an agent density function by multiplying by $N$. We then consider the dominant term or terms in the asymptotic form of $f(w) \approx -\log(P)$. From this, it is easy to derive the asymptotic redistribution function from  (\ref{eq:Redistribution2}). We will assume that if $P$ has a subquadratic decay, then $L_\infty = 1$, as was the case in the EYSM and AWM~\cite{PhysicaA, Bruce_SIAM_19}. These results are provided in Table \ref{tab:Redistribution functions}. 

\begin{table}[t]
{\footnotesize
  \caption{{ The redistribution functions for six classes of distributions that are commonly observed within the study of wealth economics. } {The redistribution function $\chi(w)$ is derived from  (\ref{eq:Redistribution2}). We have used in the case of the exponential distribution that the $\frac{D}{w}$ term must cancel due to  (\ref{eq:LargeWealthLogPw}). We also assume for higher-order Gaussian distributions that $m\in(1,\infty)$. For ease of notation, we define the constant $D$ implicitly by  $\iota(w)=B_\infty \frac{f'(w)}{w} + \frac{D}{w}$. We also define $C$ to be some normalization constant and forgo the barring of  $B_\infty$ and $\mu$. \normalfont}} \label{tab:Redistribution functions}
\begin{center}
\begin{tabular}{|c|c|c|c|} \hline
 \T Distribution&Agent Density Function &Asymptotic $f(w)$ & $\chi(w)$\\[2.5ex] \hline
\hline
   \T Exponential& $C \exp[-\lambda w]$ & $\lambda w$ & $\zeta $\\[2.5ex]
    \hline
   \T Lognormal   & $ C\frac{1}{w}\exp\big[-\frac{[\ln(w)^2-1]^2}{2\sigma^2}\big]$ & $\frac{1}{\sigma^2}\log(w)^2$ & $\zeta + \frac{D}{w}+ \frac{2B_\infty }{\sigma^2} \frac{\log(w)}{w^2} $ \\[2.5ex]
    \hline
    \T Pareto      & $C{w^{-(\alpha+1)}}$ & $(\alpha+1)\log(w)$ & $\zeta +\frac{ D}{w}+ \frac{(\alpha+1)B_\infty}{w^2}  $ \\[2.5ex]
    \hline
    \T Inverse-gamma & $Cw^{-(\alpha+1)}\exp\big[\frac{-\beta}{w}\big]$  & $(\alpha+1)\log(w) + \frac{\beta}{w}$ &   $\zeta  +\frac{ D}{w}+ \frac{(\alpha+1)B_\infty-\beta B_\infty}{w^2}$ \\[2.5ex]
    \hline
    \T \T Gaussian    & $C\exp\big[-\frac{(w-\mu)^2}{\sigma^2}\big]$ & $\frac{w^2}{\sigma^2}$ & $\lim_{w\rightarrow\infty}\chi(w) \neq \zeta $ \\[2.5ex]
    \hline
    \begin{tabular}{@{}c@{}}Higher-order \\ Gaussian\end{tabular}&  $C\exp\big[-\frac{(w-\mu)^{2m}}{\sigma^{2m}}\big]$ &  $\frac{w^{2m}}{\sigma^{2m}}$ & $\frac{2mB_\infty}{\sigma^{2m}}w^{2m-2}$\\[2.5ex]
    \hline
 \end{tabular}
\end{center}
}
\end{table}

The first four rows of Table \ref{tab:Redistribution functions} consist of subquadratic decays. Notably, each of the redistribution functions necessary to attain these wealth distribution tails differ from one another by solely a subdominant term on the order of $\frac{1}{w}$ or $\frac{1}{w^2}$. This result is all the more evident when comparing the Pareto and inverse-gamma distributions' redistribution functions at large wealth. The necessary redistribution functions for these two different decays vary by a factor of only $\frac{-\bar{B}_\infty\beta}{w^2}$. This result emphasizes the importance of the results in Section \ref{sec: Criticality}. The minute details of redistribution policy have a dynamic effect on the shape of wealth distributions in their tails when near wealth condensation criticality.

\subsection{Application to ECB data} \label{sec: ECB fittings}

In this section, we will present the fittings of the AWM with flat redistribution to the empirical wealth data provided by the ECB on the distributions of wealth of European countries~\cite{ecb_data}. The data sets we used came from the first wave of the Eurosystem Household Finance and Consumption Network's Household Finance and Consumption Survey, which provides detailed information on the wealth of over 62,000 households across fifteen Euro-area countries~\cite{ecb_data_methodology}. All households were surveyed between 2008 and 2010. The Household Finance and Consumption Survey is made available to researchers on reasonable request from the ECB. While households from Slovakia were also surveyed as a part of the Household and Consumption Survey, we did not analyze their data because the Slovakian response rate was not published in the methodological report~\cite{ecb_data_methodology}.

The Lorenz curve, denoted $\mathcal{L}(F)$, is a parametric plot of $L(w)$ against $F(w)=1-A(w)$~\cite{lorenz1905methods}. This is a curve in the unit square, parameterized by the wealth $w$, that can be shown to be concave up and lie below the diagonal. A point $(f,l)$ on the Lorenz curve tells us that a fraction $f$ of economic agents hold a fraction $l$ of wealth. The farther the Lorenz curve is from the diagonal of the unit square, the more inequality exists in an economy. Inequality can be quantified using the Gini coefficient $G$, which is defined to be two times the area under the Lorenz curve~\cite{gini1921measurement}. A  Gini coefficient of 0 would therefore represent a total oligarchy, and a Gini coefficient of 1 would represent a uniform distribution of wealth. 

First, we set $\chi(w)=\chi$ to be a constant function in this section. Constant redistribution schemes are frequently used approximations for more complicated redistribution schemes that are used in real-world economies~\cite{bisi2009kinetic}. Our goal will be to consider how close the redistribution parameter $\chi$ is to the WAA parameter $\zeta$. If these parameters are close for a given economy, this should indicate that the economy is near-critical. In this formulation, our parameter space is $\theta = \{\chi,\zeta, \lambda\}$. Let $\mathcal{L}(F)$ be the empirical Lorenz curve and $\mathcal{L}_\theta(F)$ be the theoretical (model) Lorenz curve of the agent density function given by the AWM with parameters $\theta = \{\chi,\zeta,\lambda\}$~\cite{AWM}. We define the \emph{discrepancy} by 
\begin{align}
    J(\theta)= \int_0^1 dF\; |\mathcal{L}(F)-\mathcal{L}_\theta(F)|. \label{eq:discrepancy}
\end{align}
Thus, $J(\theta)$ is the $L^1$ norm of the difference between the empirical Lorenz curve and that of the AWM with parameter choices $\theta$. The fittings to ECB data were performed by minimizing $J(\theta)$ over $\theta$. There are no guarantees for the concavity of $J(\theta)$, so we employed a global numerical search for the optimal parameters. The optimal values for $\chi$, $\zeta$, and $\lambda$ are given in Table \ref{tab:Fittings}. We let $G_\text{fit}$ refer to the Gini coefficient of the AWM-fitted Lorenz curve. 

\begin{table}[t]
{\footnotesize
  \caption{ \label{tab:Fittings} { Optimal redistribution, WAA, and negative wealth parameters for the AWM when fitted to fourteen ECB countries' empirical wealth data.} The values $\chi_\text{opt}$, $\zeta_\text{opt}$, and $\lambda_\text{opt}$ are defined to be the optimal parameters for the AWM fitting to a given country's ECB wealth data. $G_\text{fit}$ is the Gini coefficient of the corresponding wealth distribution, as obtained by the AWM~\cite{AWM}. We also report the average local error for each fitting. }
\begin{center}
\begin{tabular}{|c|c|c|c|c|c|} \hline
Country& $\chi_\text{opt}$ & $\zeta_\text{opt}$ & $\lambda_\text{opt}$ &  $G_\text{fit}$ & Average \\
& & & & & Local Error\\ \hline      \hline
 Austria       &   0.156 & 0.182   &  0.185    &  0.763      &    0.28 \% \\ \hline
      Belgium       &   1.406 & 1.514   & 0.577     &  0.589 &   0.28 \%   \\\hline
      Cyprus        &   0.164 & 0.190   & 0.096     &  0.690 &     0.23 \% \\\hline
      Germany       &   0.162 & 0.184   & 0.199     &  0.759 &   0.33 \%   \\\hline
      Spain         &   1.568 & 1.728   & 0.502     &  0.568 &   0.21 \%   \\\hline
      Finland       &   0.972 & 1.000   & 0.639     &  0.665 &     0.39 \% \\\hline
      France        &   0.556 & 0.608   & 0.286     &  0.673 &     0.48 \% \\\hline
      Greece        &   1.944 & 2.000   & 0.650     &  0.553 &     0.22 \% \\\hline
      Italy         &   1.194 & 1.300   & 0.502     &  0.601 &    0.34 \%  \\\hline
      Lithuania     &   0.896 & 1.066   & 0.425     & 0.658  &   0.31 \%   \\\hline
      Malta         &   1.154 & 1.348   & 0.377     & 0.583  &      0.18 \% \\  \hline
      Netherlands   &   1.676 & 1.516   & 0.992     & 0.647  &    0.28 \%   \\\hline
      Portugal      &0.564  &   0.678   & 0.309     & 0.672  &     0.17 \%  \\\hline
      Slovenia      & 1.978 &   1.998   & 0.618     & 0.529  &    0.15 \%  \\\hline
 \end{tabular}
\end{center}
}
\end{table}

It is notable that our data fittings of the AWM found all fourteen European countries to be near-critical (Figure \ref{fig:Chi_Zeta_plot}). The theoretical work in Section \ref{sec: Criticality} showed that when an economy is near-critical, the qualitative nature of the distribution of large wealth depends sensitively on redistribution policy. We have shown that the exact nature of the distribution of large wealth in a near-critical economy depends sensitively on the minute details of public policy, which are modeled by $\iota(w)$. Thus, it is possible that each of these countries has a qualitatively different subquadratic decay in agent density at large wealth.

\begin{figure}[htbp]
  \centering
  \includegraphics[scale=0.32]{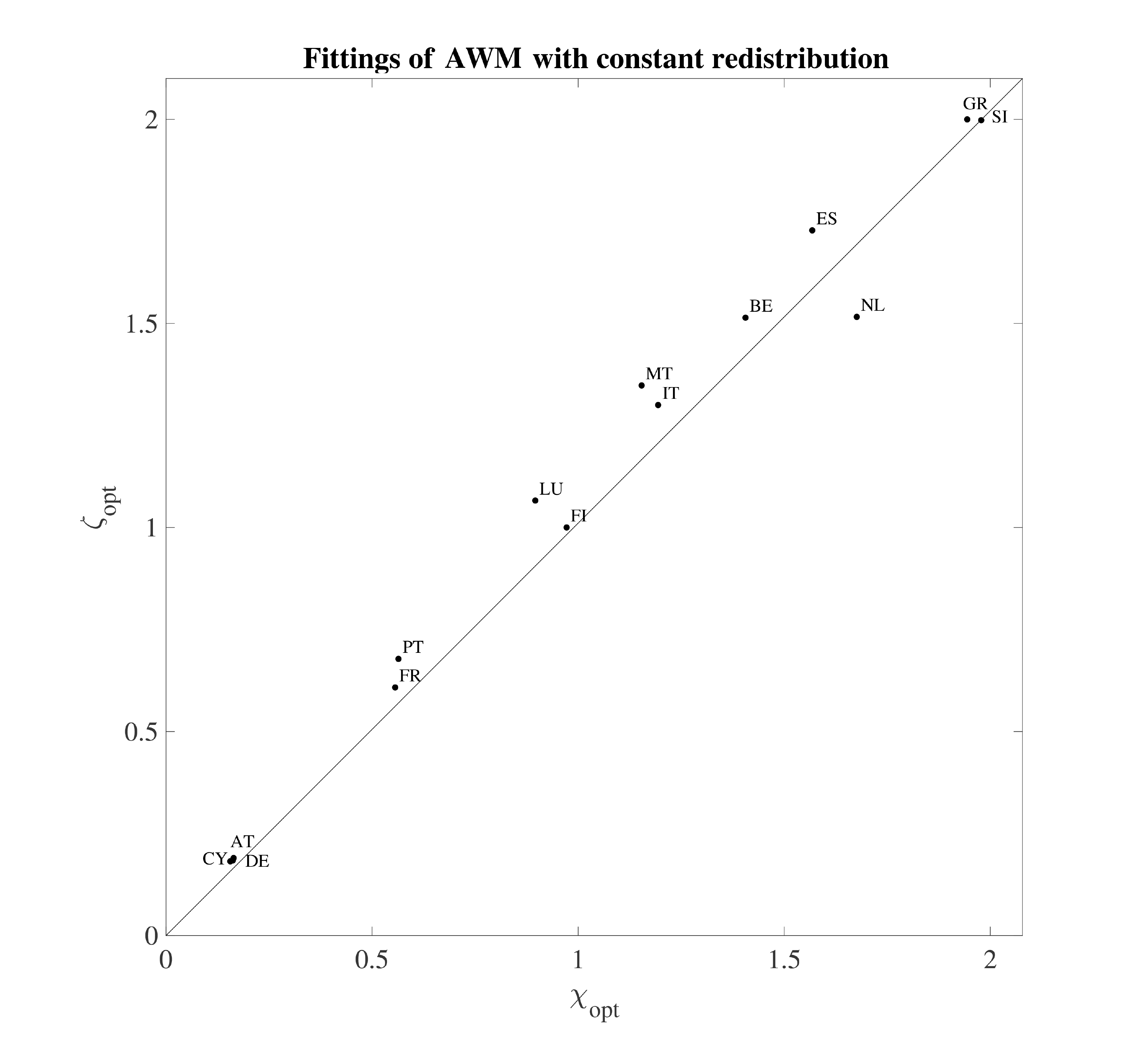}
  \caption{{\bf Plots of optimal redistribution parameters ($\chi_\text{opt}$) against optimal WAA parameters ($\zeta_\text{opt}$) for fourteen ECB countries.}}
  \label{fig:Chi_Zeta_plot}
\end{figure}

For some of the results of our fittings, the error was so small that the Lorenz curve given by our model was difficult to discriminate from that of the empirical data. For this reason, we will consider the performance of our model in terms of local error as well. We define the local error as the length of a line segment connecting the empirical data point $(f_j,l_j)$ to the model Lorenz curve, constructed so as to be perpendicular to the latter~\cite{AWM}. This section's models assumed that redistribution was constant in the AWM. Therefore, if the limit of redistribution was not exactly equal to the WAA parameter, the distribution of large wealth was assumed to be Gaussian. Despite the small average local error of these fittings, it is notable that for many countries, the vast majority of pointwise error in the fittings occurs in the tail. Four excellent examples of this relatively large error in the tail are provided in Figure \ref{fig:Fittings}. We conjecture that the pointwise error in the tail of the distribution is explained by the state of near criticality and the assumption of a Gaussian distribution at large wealth. Our work in Section \ref{sec: Criticality} shows that because these countries are near-critical, there could be a nonexponential, subquadratic decay which describes the tail of these countries distributions better than a Gaussian. 

\begin{figure}[t]
  \centering
  \includegraphics[scale=1.5]{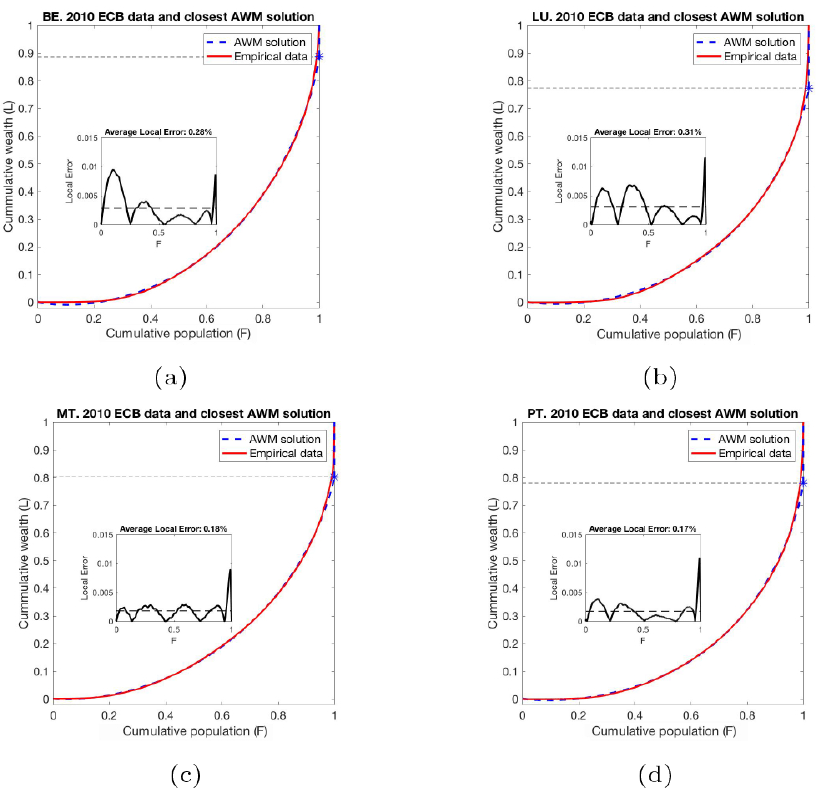}
\caption{{\bf Optimal fits of the AWM to ECB data for four countries.} For each country, we determined the parameters $(\chi, \zeta, \lambda)$ that minimize the $L^1$-norm of the difference between empirical and model Lorenz curves. Pointwise error is plotted within each figure. The fraction of oligarchical wealth can be estimated by  (\ref{eq:L_inf_AWM_specific}). Notably, the local error is relatively large in the tail, suggesting the presence of a nonexponential, subquadratic decay. (a): Belgium. (b): Luxembourg. (c): Malta. (d): Portugal.}
  \label{fig:Fittings}
\end{figure}


\section{Conclusion}\label{sec:conclusion}

This brings us to the overall conclusion of this work: that there is no universal form for most real-world economies' wealth distribution tails. We have shown that the nature of the asymptotic solution to the Fokker-Planck equation governing the AWM sensitively depends on one's choice in redistribution policy when an economy is near-critical (Section \ref{sec: Criticality}). When the AWM was fit to the wealth data of fourteen European Union countries, we found that each one was near-critical (Table \ref{tab:Fittings}). We conclude that the popular question of whether wealth decays like a Pareto distribution or an exponential distribution cannot be answered without first considering the policies of the country from which data were collected. This implies that there is no universal form for wealth distributions, at least at large wealth. Because of the exhibited sensitivity that a distribution of wealth has to the particularities of redistribution policy, we suggest a reframing of how the distribution of large wealth is studied. In particular, rather than fitting distributions to wealth data and observing how accurate those fittings may be, we suggest an emphasis on studying redistribution policy and its effects on the entire wealth distribution. This approach would more readily capture the integrodifferential nature of the equations governing wealth distributions~\cite{AWM}. The bulk of econophysics research, in tandem with the work described in our research, shows that this is a more scientific and well-posed approach to understanding a distribution of wealth.

In this work, we have extended the EYSM and AWM by generalizing redistribution to be a nearly arbitrary function of wealth~\cite{PhysicaA, AWM}. We showed that every subquadratic decay satisfying the assumptions of this work is possible by means of a progressive redistribution function with a wealth limit within a neighborhood of the WAA parameter. This extends the notion of criticality -- the phenomenon of the presence of oligarchy disappearing when constant redistribution is exactly equal to the WAA parameter -- to include a plethora of subquadratic decays other than the exponential considered by prior work~\cite{PhysicaA, AWM}. These include the lognormal and Pareto distributions. We note that the way that the redistribution function tends toward its asymptote governs the distribution of wealth. This implies that near-critical systems are extremely sensitive to the minutiae of redistribution policy. Moreover, this sensitive dependence of the nature of wealth distributions implies that the redistribution policy decisions of a society are more indicative of the distribution of large wealth than any underlying economic forces when that society is near criticality.

The fact that the asymptotic redistribution rate need not be exactly equal to the WAA parameter is of the utmost importance to policy decisions in global economics. We fit the AWM to empirical wealth data from fourteen European economies and found that all lie either just above or below criticality (Figure \ref{fig:Chi_Zeta_plot}).  Our work implies that the conversation about how large wealth is distributed may be ill-posed, as the distribution will sensitively depend on the policy decision specific to those countries. This implies that a universal distribution of wealth -- like that which has been sought for from Pareto to Piketty -- is likely a chimera. We emphasize that to understand a near-critical wealth distribution, one must thoroughly analyze large wealth redistribution policies that vary across societies. 

\section*{Acknowledgements}

For part of the time spent on this work, B.M.B. was visiting the Research and Training Center of the Central Bank of Armenia in 2017. The authors gratefully acknowledge the support and hospitality of the Central Bank of Armenia. We also would like to thank Hongyan Wang, Chengli Li, and Jie Li for their help with the model fittings presented in this work. 

\bibliographystyle{siamplain}
\bibliography{mybibtexfile} 

\begin{thebibliography}{10}

\bibitem{Angle}
{\sc J.~Angle}, {\em The surplus theory of social stratification and the size
  distribution of personal wealth}, Soc. Forces, 65 (1986), pp.~293--326.

\bibitem{Baldwin}
{\sc J.~Baldwin}, {\em Nobody Knows My Name: More Notes of a Native Son}, Dial
  Press, New York, 1961.

\bibitem{bisi2009kinetic}
{\sc M.~Bisi, G.~Spiga, and G.~Toscani}, {\em Kinetic models of conservative
  economies with wealth redistribution}, Commun. Math. Sci., 7 (2009),
  pp.~901--916.

\bibitem{Bogh2}
{\sc B.~Boghosian}, {\em \uppercase{F}okker--\uppercase{P}lanck description of
  wealth dynamics and the origin of \uppercase{P}areto's law}, Int. J. Mod.
  Phys. C, 25 (2014), p.~1441008.

\bibitem{Bogh1}
{\sc B.~M. Boghosian}, {\em Kinetics of wealth and the \uppercase{P}areto law},
  Phys. Rev. E, 89 (2014), p.~042804.

\bibitem{PhysicaA}
{\sc B.~M. Boghosian, A.~Devitt-Lee, M.~Johnson, J.~Li, J.~A. Marcq, and
  H.~Wang}, {\em Oligarchy as a phase transition: The effect of wealth-attained
  advantage in a \uppercase{F}okker--\uppercase{P}lanck description of asset
  exchange}, Physica A, 476 (2017), pp.~15--37.

\bibitem{Bogh6}
{\sc B.~M. Boghosian, M.~Johnson, and J.~A. Marcq}, {\em An \uppercase{H}
  theorem for \uppercase{B}oltzmann's equation for the yard-sale model of asset
  exchange}, J. Stat. Phys., 161 (2015), pp.~1339--1350.

\bibitem{Bouchaud}
{\sc J.-P. Bouchaud and M.~M{\'e}zard}, {\em Wealth condensation in a simple
  model of economy}, Physica A, 282 (2000), pp.~536--545.

\bibitem{bricker2017changes}
{\sc J.~Bricker, L.~J. Dettling, A.~Henriques, J.~W. Hsu, L.~Jacobs, K.~B.
  Moore, S.~Pack, J.~Sabelhaus, J.~Thompson, and R.~A. Windle}, {\em Changes in
  \uppercase{US} family finances from 2013 to 2016: \uppercase{E}vidence from
  the \uppercase{S}urvey of \uppercase{C}onsumer \uppercase{F}inances}, Fed.
  Res. Bull., 103 (2017), p.~1.

\bibitem{Burda}
{\sc Z.~Burda, D.~Johnston, J.~Jurkiewicz, M.~Kami{\'n}ski, M.~A. Nowak,
  G.~Papp, and I.~Zahed}, {\em Wealth condensation in \uppercase{P}areto
  macroeconomies}, Phys. Rev. E, 65 (2002), p.~026102.

\bibitem{Chak}
{\sc A.~Chakraborti}, {\em Distributions of money in model markets of economy},
  Int. J. Mod. Phys. C, 13 (2002), pp.~1315--1321.

\bibitem{Bruce_SIAM_19}
{\sc A.~Devitt-Lee, H.~Wang, J.~Li, and B.~Boghosian}, {\em A nonstandard
  description of wealth concentration in large-scale economies}, SIAM J. Appl.
  Math., 78 (2018), pp.~996--1008.

\bibitem{dragulescu2000statistical}
{\sc A.~Dr{\u{a}}gulescu and V.~M. Yakovenko}, {\em Statistical mechanics of
  money}, Eur. Phys. J. B., 17 (2000), pp.~723--729.

\bibitem{yakovenko}
{\sc A.~Dr{\u{a}}gulescu and V.~M. Yakovenko}, {\em Exponential and power-law
  probability distributions of wealth and income in the \uppercase{U}nited
  \uppercase{K}ingdom and the \uppercase{U}nited \uppercase{S}tates}, Physica
  A, 299 (2001), pp.~213--221.

\bibitem{during2008kinetic}
{\sc B.~D{\"u}ring, D.~Matthes, and G.~Toscani}, {\em Kinetic equations
  modelling wealth redistribution: \uppercase{a} comparison of approaches},
  Phys. Rev. E, 78 (2008), p.~056103.

\bibitem{fokker1914mittlere}
{\sc A.~D. Fokker}, {\em Die mittlere \uppercase{E}nergie rotierender
  elektrischer \uppercase{D}ipole im \uppercase{S}trahlungsfeld}, Ann. Phys.,
  348 (1914), pp.~810--820.

\bibitem{SCF}
{\sc FRB}, {\em The \uppercase{S}urvey of \uppercase{C}onsumer
  \uppercase{F}inances}.
\newblock \url{https://www.federalreserve.gov/econres/scf_2016.htm}, 2016.

\bibitem{gini1921measurement}
{\sc C.~Gini}, {\em Measurement of inequality of incomes}, Econ. J., 31 (1921),
  pp.~124--126.

\bibitem{ecb_data}
{\sc HFCS}, {\em Eurosystem's \uppercase{H}ousehold \uppercase{F}inance and
  \uppercase{C}onsumption \uppercase{S}urvey}.
\newblock
  \url{https://www.ecb.europa.eu/pub/economic-research/research-networks/html/researcher_hfcn.en.html},
  2013.

\bibitem{ecb_data_methodology}
{\sc HFCS}, {\em \uppercase{T}he \uppercase{H}ousehold \uppercase{F}inance and
  \uppercase{C}onsumption \uppercase{S}urvey: \uppercase{M}ethodological
  \uppercase{r}eport for the \uppercase{f}irst \uppercase{w}ave}, 2013.

\bibitem{ispolatov}
{\sc S.~Ispolatov, P.~L. Krapivsky, and S.~Redner}, {\em Wealth distributions
  in asset exchange models}, Eur. Phys. J. B, 2 (1998), pp.~267--276.

\bibitem{kayser2020kinetic}
{\sc K.~Kayser, D.~Armbruster, and M.~Herty}, {\em Kinetic models of
  conservative economies with need-based transfers as welfare}, Kinet. Relat.
  Models, 13 (2020), pp.~169--185.

\bibitem{kayser2017kinetic}
{\sc K.~Kayser, D.~Armbruster, and C.~Ringhofer}, {\em Kinetic models of
  need-based transfers}, in International Conference on Applied Mathematics,
  Modeling and Computational Science, New York, 2017, Springer, pp.~521--530.

\bibitem{kolmogoroff1931analytischen}
{\sc A.~Kolmogoroff}, {\em {\"U}ber die analytischen \uppercase{M}ethoden in
  der \uppercase{W}ahrscheinlichkeitsrechnung}, Math. Ann., 104 (1931),
  pp.~415--458.

\bibitem{AWM}
{\sc J.~Li, B.~M. Boghosian, and C.~Li}, {\em The affine wealth model: An
  agent-based model of asset exchange that allows for negative-wealth agents
  and its empirical validation}, Physica A, 516 (2019), pp.~423--442.

\bibitem{lima2020nonlinear}
{\sc H.~Lima, A.~R. Vieira, and C.~Anteneodo}, {\em Nonlinear redistribution of
  wealth from a \uppercase{F}okker-\uppercase{P}lanck description},
  ar\uppercase{X}iv preprint ar\uppercase{X}iv:2007.11680,  (2020).

\bibitem{lindert2017rise}
{\sc P.~H. Lindert}, {\em The \uppercase{r}ise and \uppercase{f}uture of
  \uppercase{p}rogressive \uppercase{r}edistribution}, Commitment to Equity
  (CEQ) Institute Working Paper, 73 (2017).

\bibitem{lorenz1905methods}
{\sc M.~O. Lorenz}, {\em Methods of measuring the concentration of wealth},
  Publ. Amer. Statist. Assoc., 9 (1905), pp.~209--219.

\bibitem{moukarzel2007wealth}
{\sc C.-F. Moukarzel, S.~Gon{\c{c}}alves, J.-R. Iglesias,
  M.~Rodr{\'\i}guez-Achach, and R.~Huerta-Quintanilla}, {\em Wealth
  condensation in a multiplicative random asset exchange model}, Europ. Phys.
  J., 143 (2007), pp.~75--79.

\bibitem{planck1917satz}
{\sc V.~Planck}, {\em {\"U}ber einen \uppercase{S}atz der statistischen
  \uppercase{D}ynamik und seine \uppercase{E}rweiterung in der
  \uppercase{Q}uantentheorie}, \uppercase{S}itzungberichte der
  \uppercase{K}{\"oniglich-\uppercase{P}reussichen \uppercase{A}kademie der
  \uppercase{W}issenchaften physikalisch-mathematischen},  (1917),
  pp.~324--341.

\bibitem{Risken}
{\sc H.~Risken}, {\em \uppercase{F}okker-\uppercase{P}lanck equation}, in The
  \uppercase{F}okker-\uppercase{P}lanck Equation, Springer, New York, 1996.

\bibitem{Rosen}
{\sc M.~N. Rosenbluth, W.~M. MacDonald, and D.~L. Judd}, {\em
  \uppercase{F}okker-\uppercase{P}lanck equation for an inverse-square force},
  Phys. Rev., 107 (1957), pp.~1--6.

\bibitem{silva2004temporal}
{\sc A.~C. Silva and V.~M. Yakovenko}, {\em Temporal evolution of the
  ``thermal'' and ``superthermal'' income classes in the \uppercase{USA} during
  1983--2001}, Europhy. Lett., 69 (2004), pp.~304--310.

\bibitem{toscani2009wealth}
{\sc G.~Toscani}, {\em Wealth redistribution in conservative linear kinetic
  models}, Europhys. Lett., 88 (2009), p.~10007.

\bibitem{toscani2010wealth}
{\sc G.~Toscani and C.~Brugna}, {\em Wealth redistribution in
  \uppercase{B}oltzmann-like models of conservative economies}, in Econophysics
  and Economics of Games, Social Choices and Quantitative Techniques, Springer,
  New York, 2010, pp.~71--82.

\bibitem{uhlenbeck1930theory}
{\sc G.~E. Uhlenbeck and L.~S. Ornstein}, {\em On the theory of the
  \uppercase{B}rownian motion}, Phys. Rev., 36 (1930), p.~823.

\bibitem{Pareto2}
{\sc V.~\uppercase{P}areto}, {\em Cours d'Economie Politique: Nouvelle edition
  par G.-H. Bousquet et G. Busino}, Librairie Droz, Geneva, Switzerland, 1964.

\bibitem{Pareto}
{\sc V.~\uppercase{P}areto}, {\em La Courbe de la repartition de la richesse},
  Imprimerie Ch. Viret-Genton, Geneva, Switzerland, 1965.

\bibitem{sheng}
{\sc S.~Xu and B.~M. Boghosian}, {\em Time-dependent solutions to the affine
  wealth model for the transport of agents and wealth on the wealth spectrum},
  To Appear,  (2021).

\end{thebibliography}

\appendix

\section{Class of functions satisfying our assumptions}\label{sec: satisfactory functions}

In this section, we prove that given a function $f(w)$ satisfying some loose conditions, the assumptions of this work will hold.

\begin{proposition} \label{prop: Assumption Class}
    If $f(w) \approx w^p \log^q(w)$, where $p>0,q\in\mathbb{R}$ or $p=0,q>1$, (\ref{eq:Assumption1})-(\ref{eq:Assumption3}) are satisfied. 
\end{proposition}
\begin{proof}
    {Case 1}; Assume that $p>0$ and $q\in\mathbb{R}$: 
        \begin{enumerate}
        \item By assumption, $f'(w) \approx ap\; w^{p-1}\log^q(w)$, which is positive for all large wealth. Therefore,  (\ref{eq:Assumption1}) holds.
        
        \item Consider the following limit:
        \begin{align*}
            \lim_{w\rightarrow\infty} \frac{f''(w)}{[f'(w)]^2} &=\frac{1}{a} \lim_{w\rightarrow\infty }\frac{p(p - 1)  \log^2(w) +(2 p - 1) q \log(w) +  q(q - 1) )}{w^p\log^q(w)[p\log(w)+q]^2}\\ &= \frac{p-1}{ap} \lim_{w\rightarrow\infty }\frac{1}{\log^q(w)w^p}\\ &= 0.
        \end{align*}
        Therefore,  (\ref{eq:Assumption2}) holds. 
        
        \item Consider the following limit:
        \begin{align*}
            \lim_{w\rightarrow\infty} \frac{w^2}{[f'(w)]^2}e^{-f(w)} &= \lim_{w\rightarrow\infty } \frac{\log^{2-2q}(w)e^{-aw^p\log^q(w)}}{a^2w^{2(p-2)}(p\log(w)+q)^2}\\ &=\frac{1}{a^2p^2}\lim_{w\rightarrow\infty }\frac{e^{-aw^p\log^q(w)}}{ w^{2(p-2)}\log^{2q}(w)}\\
            &=\frac{1}{a^2p^2}\lim_{w\rightarrow\infty }\frac{w^{-aw^p\log^{q-1}(w)}}{ w^{2(p-2)}\log^{2q}(w)}\\ &= 0.
        \end{align*}
        Therefore,  (\ref{eq:Assumption3}) holds.
    \end{enumerate}
    
    {Case 2}; Assume that $p=0$ and $q>1$:
    \begin{enumerate}
        \item By assumption, $f'(w)  \approx q\frac{\log^{q-1}(w)}{w}$,  which is positive for all large wealth. Thus,  (\ref{eq:Assumption1}) holds.
        
        \item Consider the following limit:
        \begin{align*}
            \lim_{w\rightarrow\infty} \frac{f''(w)}{[f'(w)]^2} = \lim_{w\rightarrow\infty} \bigg[\frac{q-1}{q\log^q(w)} -\frac{1}{\log^{q-1}(w)}\bigg] = 0.
        \end{align*}
        Thus,  (\ref{eq:Assumption2}) holds.
        
        \item Consider the following limit:
        \begin{align*}
            \lim_{w\rightarrow\infty}\frac{w^2}{[f'(w)]^2}e^{-f(w)} &=\lim_{w\rightarrow\infty} \frac{w^4}{q^2\log^{2(q-1)}}\exp[-\log^q(w)]\\
            &= \lim_{w\rightarrow\infty }\frac{w^{4-\log^{q-1}(w)}}{q^2 \log^{2(q-1)}(w)}\\&= 0.
        \end{align*}
        Thus,  (\ref{eq:Assumption3}) holds.
        
    \end{enumerate}
    Therefore, our assumptions are satisfied by the following class of functions: \newline $w^p\log^q(w)$, where either $p>0$ and $q\in\mathbb{R}$ or $p=0$ and $q>1$. 
\end{proof}

\section{Further notes on the consequences of our asymptotic assumptions}\label{sec: consequences of assumptions}

\hspace{10pt} In this section, we derive some facts which are used throughout this work. Lemmas \ref{lemma: A2}-\ref{lemma: A4} are consequences of our assumptions: (\ref{eq:Assumption1})-(\ref{eq:Assumption3}). We assume that wealth is sufficiently large that these assumptions are valid.
\begin{lemma}\label{lemma: A1}
    If $\lim_{w\rightarrow\infty} f(w) = \infty$ and $g'(w)\gg f'(w)$, then $g(w) \gg f(w)$.
\end{lemma}
\begin{proof}
    Because $g'(w) \gg f'(w)$ for all positive $M$, however large, there exists an $n_M>0$, where $w>n_M$ implies that $g'(w) > Mf'(w)$. This implies that
    \begin{align*}
        g(w) &= g(n_M) + \int_{n_M}^w dx\; g'(x),\\
        g(w) &> g(n_M) + M\int_{n_M}^w dx\; f'(x),\\
        g(w) &> [g(n_M)-Mf(n_M)] + Mf(w),\\
        \frac{g(w)}{f(w)} &> \frac{g(n_M)}{f(w)}- M\frac{f(n_M)}{f(w)} + M.
    \end{align*}
    Note that since we assume that $n_M$ is a constant and that $f(w)$ is an infinite function, for $w$ large, $\frac{f(n_M)}{f(w)}\approx \frac{g(n_M)}{f(w)}\approx 0$. Then at large wealth, we have that there exists an $n_M'>0$ such that for all $w>n_M'$, it is true that $g(w) >Mf(w)$. Thus, $g(w)\gg f(w)$.
\end{proof}

\begin{lemma}\label{lemma: A2}
    Under the assumptions listed in (\ref{eq:Assumption1})-(\ref{eq:Assumption3}), it can be proven that $\exp[f(w)]\gg \frac{w^2}{f'(w)}$. 
\end{lemma}
\begin{proof}
By  (\ref{eq:Assumption3}), $e^{f(w)}\gg \frac{w^2}{[f'(w)]^2}$. Since $\frac{w^2}{[f'(w)]^2}\gg \frac{w}{f'(w)}$, we have through transitivity that $e^{f(w)}\gg \frac{w}{f'(w)}$. Thus, 
\begin{align*}
    e^{f(w)}&\gg \frac{w}{f'(w)},\\
    f'(w)e^{f(w)} &\gg w.
\end{align*}
Note that because $f'(w)e^{f(w)}\gg w$, both $f'(w)$ and $e^{f(w)}$ are necessarily infinite functions of wealth. Therefore, the assumptions of Lemma \ref{lemma: A1} are satisfied, and by integrating, we see that
\begin{align}
    e^{f(w)}\gg w^2.\label{eq:w_squared}
\end{align}
Taking the geometric mean of  (\ref{eq:Assumption3}) and  (\ref{eq:w_squared}), we arrive at our result.
\begin{align*}
    e^{f(w)} = \sqrt{e^{2f(w)}} \gg \sqrt{w^2 \frac{w^2}{[f'(w)]^2}} = \frac{w^2}{f'(w)}.
\end{align*}
\end{proof}

\begin{lemma}\label{lemma: A3}
    Under the assumptions listed in (\ref{eq:Assumption1})-(\ref{eq:Assumption3}), it can be proven that $\exp[f(w)]\gg w$. 
\end{lemma}
\begin{proof}
By  (\ref{eq:w_squared}) in our proof of Lemma \ref{lemma: A2}, $\exp[f(w)]\gg w^2\gg w$. 
\end{proof}

\begin{lemma}\label{lemma: A4}
    Under the assumptions listed in (\ref{eq:Assumption1})-(\ref{eq:Assumption3}), it can be proven that $f'(w)\gg \frac{1}{w}$. 
\end{lemma}
\begin{proof}
    By Lemma \ref{lemma: A3}, 
    \begin{align*}
        \exp[f(w)] &\gg w,\\
        f(w) &\gg \log(w).
    \end{align*}
    In particular,
    \begin{align*}
      0=  \lim_{w\rightarrow\infty}\frac{\log(w)}{f(w)} = \lim_{w\rightarrow\infty}\frac{\frac{1}{w}}{f'(w)} ,
    \end{align*}
    where we have applied L'H\^{o}pital's rule to arrive at our conclusion. 
\end{proof}

\section{Reduction of the steady-state EYSM Fokker-Planck equation}\label{sec: Approximations}

\hspace{10pt} In this appendix, we will derive  (\ref{eq:LargeWealthLogPw}) using the assumptions provided in (\ref{eq:Assumption1})-(\ref{eq:Assumption3}). First, let us rearrange the steady-state EYSM Fokker-Planck equation with general redistribution --  (\ref{eq:FokkerPlanckSteadyState}):
\begin{align}
    \frac{d}{dw}[\log P] &= \frac{\frac{T}{N}-\chi(w)w - \zeta\big[\frac{2}{\mu}\big(B-\frac{w^2}{2}A\big) + (1-2L)w\big] - wAP}{B+\frac{w^2}{2}A} \label{eq:Log_ODE}.
\end{align}
By  Lemma \ref{lemma:ParetoPotentialApproximation}, we can approximate the Pareto-Lorenz potentials on $(M,\infty)$ by the following functional forms:
\begin{align}
    A(w) = \frac{1}{N}\int_w^\infty dx\;P(x) &\approx \frac{1}{N}\frac{1}{f'(w)}P(w), \label{eq:AwLargeWealth}\\
    L(w) = L_\infty - \frac{1}{W}\int_w^\infty dx\; xP(x) &\approx L_\infty - \frac{1}{W}\frac{w}{f'(w)}P(w),\label{eq:LwLargeWealth}\\
    B(w) = B_\infty - \frac{1}{N}\int_w^\infty dx\;\frac{x^2}{2}P(x) &\approx B_\infty -\frac{1}{2N}\frac{w^2}{f'(w)}P(w).\label{eq:BwLargeWealth}
\end{align}
The process by which  (\ref{eq:Log_ODE}) reduces to  (\ref{eq:LargeWealthLogPw}) is rather arduous and requires the application of the assumptions listed in (\ref{eq:Assumption1})-(\ref{eq:Assumption3}) as well as their consequences, which are provided in Section \ref{sec: consequences of assumptions}. For this reason, we will prove three lemmas (Lemmas \ref{lemma: B1}-\ref{lemma: B3}). Then, in Proposition \ref{prop: A1}, we will show that these three lemmas can be used to derive  (\ref{eq:LargeWealthLogPw}) as a large-wealth approximation of $P(w)$. 

\begin{lemma}\label{lemma: B1}
At large wealth,
    \begin{align*}
        B \pm \frac{w^2}{2}A \approx B_\infty.
    \end{align*}
\end{lemma}
\begin{proof}
Consider this quantity in the context of our asymptotic approximations of the Pareto-Lorenz potentials,
\begin{align*}
    B \pm \frac{w^2}{2}A &\approx B_\infty - \frac{1}{2N}\frac{w^2}{f'(w)}P(w) \pm \frac{1}{2N}\frac{w^2}{f'(w)}P(w)\\
    &\approx B_\infty +\frac{C(-1\pm 1)}{2N}\frac{w^2}{f'(w)}e^{-f(w)}\\
    &= \bigg[B_\infty e^{f(w)}+\frac{C(-1\pm1)}{2N}\frac{w^2}{f'(w)}\bigg]e^{-f(w)}\\
    &\approx B_\infty,
\end{align*}
where we have used that $e^{f(w)}\gg\frac{w^2}{f'(w)}$, which is true by Lemma \ref{lemma: A2}.
\end{proof}

\begin{lemma}\label{lemma: B2}
    At large wealth, for any constant $c$,
    \begin{align*}
        wAP \ll c.
    \end{align*}
\end{lemma}
\begin{proof} We will prove the equivalent condition that $c + wAP \approx c$ at large wealth:
\begin{align*}
     c+ wAP = c + \frac{C}{N} we^{-f(w)} \int_0^w dx\; P(x) &= \bigg[c e^{f(w)} + \frac{C}{N} w A(w) \bigg]e^{-f(w)}.
\end{align*}
Recall that $A(w)$ is bounded by $[0,1]$ for all $w$, and our assumptions implied that $e^{f(w)}\gg w$, which was proven in Lemma \ref{lemma: A3}. Then
\begin{align*}
c\leq \bigg[c e^{f(w)} + \frac{C}{N} w A(w) \bigg]e^{-f(w)}&\leq  \bigg[c e^{f(w)} + \frac{C}{N} w \bigg]e^{-f(w)}\approx c.
\end{align*}
Therefore, $c + wAP \approx c$.
\end{proof}

\begin{lemma}\label{lemma: B3}
    At large wealth,
    \begin{align*}
        \zeta\bigg[\frac{2}{\mu}\bigg(B-\frac{w^2}{2}A\bigg) + (1-2L)w\bigg] &\approx \zeta\bigg[\frac{2B_\infty}{\mu} + (1-2L_\infty)w \bigg].
    \end{align*}
\end{lemma}
\begin{proof}
Consider this quantity in the context of our asymptotic approximations of the Pareto-Lorenz potentials, stated in (\ref{eq:AwLargeWealth})-(\ref{eq:BwLargeWealth}). Applying Lemma \ref{lemma: B1},
\begin{align*}
\begin{split}
    \zeta\bigg[\frac{2}{\mu}\bigg(B-\frac{w^2}{2}A\bigg) + (1-2L)w\bigg] \approx& \zeta\bigg[\frac{2B_\infty}{\mu} + (1-2L_\infty)w+ \frac{2C}{W}\frac{w^2}{f'(w)}e^{-f(w)}\bigg]
\end{split}\\
\approx& \zeta\bigg[\frac{2B_\infty}{\mu} + (1-2L_\infty)w \bigg],
\end{align*}
where we have used that $e^{f(w)}\gg \frac{w^2}{f'(w)}$, which was proven in Lemma \ref{lemma: A2}.
\end{proof}

We are now equipped to use Lemmas \ref{lemma: B1}-\ref{lemma: B3} to prove that $f(w)=-\log[P(w)]$ reduces to the functional form provided in  (\ref{eq:LargeWealthLogPw}).
\begin{proposition}
At large wealth,
\begin{align*}
    f(w) &\approx \frac{1}{B_\infty}\int^w dx\; \chi(x)x +\frac{\zeta(1 - 2L_\infty)}{2B_\infty}w^2+ \frac{2\zeta B_\infty  - \frac{T}{N}\mu }{B_\infty \mu}w.
\end{align*}\label{prop: A1}
\end{proposition}
\begin{proof}
By Lemmas \ref{lemma: B1}-\ref{lemma: B3},
\begin{align}
    \frac{d}{dw}[\log(P)]&= \frac{\frac{T}{N}-\chi(w)w - \zeta\big[\frac{2}{\mu}\big(B-\frac{w^2}{2}A\big) + (1-2L)w\big] - wAP}{B+\frac{w^2}{2}A}\\ &\approx  
    \frac{\frac{T}{N}-\chi(w)w - \zeta\big[\frac{2B_\infty}{\mu} + (1-2L_\infty)w \big] }{B_\infty}\nonumber \\
        &= \frac{1}{B_\infty}\bigg(\frac{T}{N}-\chi(w)w\bigg)  - \zeta\bigg[\frac{2}{\mu}+ \frac{1}{B_\infty}(1-2L_\infty)w\bigg].
\end{align}
Note that 
\begin{align*}
    \frac{d}{dw}\big(\log P\big)= \frac{d}{dw}[-f(w)] = -f'(w).
\end{align*}
Therefore, we find that the tail of the distribution satisfies
\begin{align*}
    -f'(w)&\approx  -\frac{1}{B_\infty}\chi(w)w -\frac{\zeta(1 - 2L_\infty)}{B_\infty}w +\frac{2\zeta B_\infty  - \frac{T}{N}\mu }{B_\infty \mu},\\
    f(w) &\approx \frac{1}{B_\infty}\int^w dx\; \chi(x)x +\frac{\zeta(1 - 2L_\infty)}{2B_\infty}w^2+ \frac{2\zeta B_\infty  - \frac{T}{N}\mu }{B_\infty \mu}w,
\end{align*}
where we have disregarded a subdominant constant of integration.  Thus, we have arrived at our result. 

\end{proof}

\end{document}